\title{Local Correction of Juntas}
\author{
Noga Alon
\thanks{
Sackler School of Mathematics and Blavatnik School of Computer Science,
Tel Aviv University, Tel Aviv 69978, Israel and Institute for Advanced
Study, Princeton, New Jersey 08540, USA.  Email: nogaa@tau.ac.il.
Research supported in part by an ERC Advanced grant, by a USA-Israeli
BSF grant and by NSF grant No.  DMS-0835373.
}
\and
Amit Weinstein
\thanks{
Blavatnik School of Computer Science, Tel Aviv University, Tel Aviv 69978,
Israel. Email: amitw@tau.ac.il. Research supported in part by an 
ERC advanced grant.
}
}

\documentclass [11pt]{article}
\usepackage[]{amsfonts,amsmath,amssymb,amsthm}
\usepackage[usenames]{color}

\newtheorem{theo}{Theorem}
\newtheorem{prop}[theo]{Proposition}

\newtheorem{coro}[theo]{Corollary}

\newtheorem*{question*}{Question}

\newtheorem{definition}{Definition}
\newtheorem*{definition*}{Definition}

\newcommand{\Exp}[1]{\mathbf{E}\left[#1\right]}

\newcommand{\ZZ}{\mathbb{Z}}

\newcommand{\DDD}{\mathcal{D}}
\newcommand{\Inf}{\mathrm{Inf}}
\newcommand{\Maj}{\mathrm{Maj}}

\def \eps {\varepsilon}

\begin{document}
\maketitle

\begin{abstract}
A Boolean function $f$ over $n$ variables is said to be $q$-locally
correctable if, given a black-box access to a function $g$ which is
"close" to an isomorphism $f_{\sigma}$ of $f$, we can compute
$f_{\sigma}(x)$ for \emph{any} $x \in \ZZ_2^n$ with good probability
using $q$ queries to $g$.

We observe that any $k$-junta, that is, any 
function which depends only on
$k$ of its input variables, is $O(2^k)$-locally correctable.
Moreover, we show that
there are examples where this is essentially best possible, and
locally correcting some $k$-juntas requires a number of queries which is
exponential in $k$. These examples, however,
are far from being typical, and indeed we prove that for
almost every $k$-junta, $O(k \log k)$ queries suffice.
\end{abstract}

\section{Introduction}
The field of property testing of Boolean functions received a considerable
amount of attention in the last two decades. Many properties of
functions have been examined in order to estimate what is the needed
query complexity for testing them, that is, the number of inputs of
the function one has to read in order to distinguish between a function
that satisfies the property and one that is "far" from satisfying it. In
particular, one is usually interested in properties for which the number
of queries is independent of the input size. Some of these properties
are linearity \cite{BLR}, being a dictator function \cite{BGS,PRS},
a junta \cite{CG,FKRSS}, or a low-degree polynomial \cite{AKKLR}.

Another property that one might consider testing is functions isomorphism,
i.e.  testing if two functions are identical up to relabeling of the input
variables.  A common scenario is where one function is given in advance
and the goal of the tester is to determine if the second input function
is isomorphic to it or far from any isomorphism of it. Several recent
results indicate that testing this property is hard for most functions
(requires $\Omega(n)$ queries), and specifically for $k$-juntas there
are lower bounds which depend on $k$ (see e.g. \cite{BO,AB,CGM}).

The focus of our work is not testing such properties, but rather locally
correcting functions, that is, determining the value of a function in a
given point by reading its values in several other points.
This is closely related to random self-reducibility, as pointed out already in \cite{BLR}.
More precisely, we care about locally correcting specific functions which are known up
to isomorphism.

\begin{question*} Given a specific Boolean function $f$, what is the needed
query complexity in order to correct an input function which is
\emph{close} to some isomorphism $f_{\sigma}$ of $f$?
\end{question*}

This question can be seen as a special case of locally correctable 
codes (see,
e.g., \cite{Y}). Namely, each codeword would be the $2^n$ evaluations of an
isomorphism of the input function (at most $n!$ distinct codewords) and we
would like to correct 
any specific value of the given noisy codeword using as
few queries as possible.

Here we study the above question mostly for juntas. We provide both
lower and upper exponential bounds for the query complexity of locally
correcting juntas with 
respect to their size. However, the given lower bound is
applicable only to a small portion of the juntas and in fact 
we show that most
$k$-juntas are locally correctable using a nearly linear (in $k$) number of
queries.

\subsection{Preliminaries}

In order to correct functions, we need to first define when two
functions are "close", as otherwise correction is hopeless. We use
the common definition, saying two Boolean functions are $\eps$-close
if they agree on all but at most an $\eps$ fraction of the inputs. The
following definition best describes the focus of our work,
indicating when a function is locally correctable.

\begin{definition*}
A Boolean function $f: \ZZ_2^n \mapsto \ZZ_2$ is said to be
$q$-locally correctable for $\eps > 0$ if the following holds. There
exists an algorithm that given an input function $g$ which is
$\eps$-close to an isomorphism $f_{\sigma}$ of $f$, can determine
the value $f_{\sigma}(x)$ for \emph{any} specific $x \in \ZZ_2^n$
with probability at least $2/3$, using $q$ queries to $g$.
\end{definition*}

More generally, we also define a family of functions to be locally correctable
when we do not require to know which specific function from the family we are trying to correct.

\begin{definition*}
A family $\mathcal{F}$ of Boolean functions is said to be $q$-locally correctable for
$\eps > 0$ if the following holds.
There exists an algorithm that given an input function $g$ which is
$\eps$-close to an isomorphism $f_{\sigma}$ of $f$,
for some $f \in \mathcal{F}$ from the family, can determine
the value $f_{\sigma}(x)$ for \emph{any} specific $x \in \ZZ_2^n$
with probability at least $2/3$, using $q$ queries to $g$.
\end{definition*}

A crucial observation when looking at the above definitions is the fact
that the mentioned algorithm must hold for \emph{every} input $x \in
\ZZ_2^n$. Replacing this requirement by the ability to determine the
value at a uniform random $x$, any function would be trivially 1-locally
correctable for $\eps \leq 1/3$ as at least $2/3$ of the inputs remain
unmodified. In addition, it is useful to think of $\eps$ as a constant
independent of $n$, which however can depend on some property of the
function $f$. This dependence is often required to ensure that $g$
is close to a unique isomorphism of $f$ (up to equivalent isomorphisms).

A simple result regarding juntas is an exponential upper bound for the
number of queries, in terms of the junta's size. For this upper bound,
we use the analysis of testing low-degree polynomials and therefore get
the following more general bound.

\begin{prop}
Every polynomial of degree $k$ is $O(2^k)$-locally correctable for $\eps <
2^{-k-3}$.
\end{prop}

\begin{proof}[Proof sketch.]
The techniques used in testing low-degree polynomials rely on their
values on the points of random affine subcubes inside $\ZZ_2^n$ which
are defined by random bases of $k+1$ vectors and an offset in $\ZZ_2^n$
(see \cite{AKKLR}). Taking such a subcube and evaluating the sum of a
degree $k$ polynomial on all $2^{k+1}$ elements of it, always results in
zero. The test itself selects several such random subcubes and verifies
that this is indeed the case.  Since in our case, we are given some
specific input $x \in \ZZ_2^n$ for which we want to correct the function,
we can use a similar argument.

Given the input $x$, we randomly select $k+1$ vectors $x_1, x_2, \ldots,
x_{k+1}$ and consider the affine subcube whose basis is the set of 
these $k+1$
vectors and whose offset is $x$. Since the sum of evaluations inside this
affine subcube (which includes $x$) is zero, we can deduce the value at
$x$ by querying the other $2^{k+1}-1$ locations of the cube, assuming
none of them was modified. Since we select the $k+1$ vectors in the basis
randomly, relying on the (easy case in the) analysis of \cite{AKKLR}
which is based on the fact that each input queried is uniformly random,
we can bound this probability by $(2^{k+1}-1)\eps < 1/4$ and therefore
this algorithm indicates that $f$ is indeed $O(2^k)$-locally correctable
for $\eps < 2^{-k-3}$.
\end{proof}

\begin{coro}
The family of $k$-juntas is $O(2^k)$-locally correctable for $\eps < 2^{-k-3}$.
\end{coro}

\begin{proof} A $k$-junta is in particular a polynomial of degree $k$ and
therefore is also $O(2^k)$-locally correctable using the above proposition.
In addition, the algorithm suggested by the proposition does not require any
knowledge about the input function except for it being a polynomial
of degree $k$, thus
the family of $k$-juntas is $O(2^k)$-locally correctable.
\end{proof}

\subsection{Our results}

A natural question is whether the exponential upper bound for low-degree
polynomials, which is applicable also for juntas, is indeed tight
in the case of juntas.  We show that the answer is indeed positive,
but only for a small fraction of the juntas. In other words, for some
juntas the exponential upper bound is also best possible but this is
far from being the typical case.

\begin{theo}\label{theo-junta-lower}
There exist some $k$-juntas which require $2^{\Omega(k)}$ (adaptive or
non-adaptive) queries in order to be 
locally corrected, even for $\eps$ which is
exponentially small in $n$.
\end{theo}

In the typical case, however, i.e., for almost every junta, the lower
bound above is far from being tight and in fact one can correct a typical
$k$-junta using a nearly linear number of queries (in $k$). Formally,
we prove the following.

\begin{theo}\label{theo-junta-upper} A $k$-junta 
in which every influencing variable has influence of at least $1/50$
is $O(k \log k)$-locally correctable 
for $\eps < 2^{-k-3}$. Therefore this is
the case for almost every $k$-junta.
\end{theo}

\begin{coro}\label{coro-junta-upper}
The family of $k$-juntas in which every influencing variable has influence of
at least $1/50$ is $O(k \log k)$-locally correctable for $\eps < 2^{-k-3}$.
\end{coro}

\section{Local correction of $k$-juntas}\label{sec-junta-proofs}

We start this section with the proof of the lower bound for some juntas. The
juntas used in the proof are very sparse, having an exponentially small
fraction of inputs for which the value of the function is 1.

\begin{proof}[Proof of Theorem \ref{theo-junta-lower}]
Given $k < n \in \mathbb{N}$ where $n$ is even, define $f$ to be the AND
function of the first $k$ literals $x_1, \ldots, x_k$. In order to prove a
lower bound for 
the number of queries, we  use Yao's principle. To this end, we
define two distributions on functions which are all 
$o(1)$-close to being isomorphic to $f$,
one for which the algorithm should return zero and
another for which the algorithm should return one (denoted by $\DDD_0$ and
$\DDD_1$ respectively). 
We further show that any algorithm that performs only
$2^{o(k)}$ queries would not be able to distinguish between the two
distributions with non-negligible probability.

We first describe the distribution $\DDD_0$ as follows. We randomly choose a
permutation $\sigma \in S_n$ so that $\sigma(i) \in [n/2]$ for every $i \in
[k]$, meaning the $k$ relevant variables are all in the first half.\footnote{
Throughout this work we use the notation $[\ell] := \{1,2,\ldots, \ell\}$. }
The function $g$ given to the algorithm is defined by $g(y) = f_{\sigma}(y)$
whenever the Hamming weight of $y$ is at most $0.3n$ in each half (i.e.
$\sum_{i=1}^{n/2} y_i \leq 0.3n$ and $\sum_{i=n/2+1}^{n} y_i \leq 0.3n$) and
otherwise $g(y) = 0$. Notice that indeed $g$ is $o(1)$-close to being isomorphic to $f$
as we modified only an $o(1)$-fraction of the inputs.
The input $x$ is set to be the balanced input of $n/2$
zeros followed by $n/2$ ones. 
Clearly $f_\sigma(x) = 0$ for every instance in
$\DDD_0$ as required.

The distribution $\DDD_1$ is similar to $\DDD_0$ with one
modification. The permutation $\sigma$ is chosen so that $\sigma(i)
\not \in [n/2]$ for every $i \in [k]$. The choice of $x$ and the
locations where we fix $g(y) = 0$ are defined as before and indeed
$f_{\sigma}(x) = 1$ for every instance in $\DDD_1$.

We first show that an arbitrary query to $g$ in either distribution
would output one with probability at most $2^{-\Omega(k)}$. Let $y$
be some query the algorithm performs. Clearly, if the Hamming weight
of $y$ in either half is more than $0.3n$, the result would be zero
according to the definition of $g$ in both distributions. Otherwise,
the probability that $g(y) = 1$ is given by
$$
{m \choose k} / {n/2 \choose k} = \frac {(m-k+1)(m-k+2)\cdots m}
{(n/2-k+1)(n/2-k+2)\cdots (n/2)} \leq \left( \frac {3n/10} {n/2}
\right)^k = 0.6^k
$$
where $m$ is the 
Hamming weight of $y$ in the relevant half (either the first
half for $\DDD_0$ or the second half for $\DDD_1$), which is known to be at
most $0.3n$. Therefore, 
any algorithm that performs at most $2^{o(k)}$ queries
would find a $y$ 
for which $g(y) = 1$ only with negligible probability, and it
would not be able to distinguish between $\DDD_0$ 
and $\DDD_1$ with noticeable probability. Notice that the
proof implies that 
using an adaptive algorithm would not yield any improvement
as we can predict all results to be zero in advance (and therefore this is
equivalent to a non-adaptive algorithm).
\end{proof}

The fact that the AND junta is very sparse was crucial for the above
proof. In order to prove a better upper bound for most juntas, we
need some restriction that would ensure the function is far from being
sparse. In Theorem \ref{theo-junta-upper} we required something even
stronger, that the influence of every influencing variable, that
is, any
of the $k$ special variables of the junta, is at least $1/50$.

\begin{definition}[Influence]
Given a Boolean function $f: \ZZ_2^n \mapsto \ZZ_2$, the
\emph{influence} of $i$ with respect to $f$ is defined by
$$
\Inf_i(f) = \Pr_x \left[ f(x) \neq f(x + e_i) \right]
$$
where $e_i$ is the vector having 1 only at location $i$.
\end{definition}

Thus, the influence is the probability that changing the value of
the $i$th variable will also change the value of the function. This
probability is taken over all values of $x$, and is therefore the expected
influence of $i$ in a restricted function (when the variable $i$ itself
is not restricted).  Moreover, if the influence of some variable $i$
is greater than $1/50$, then the function is $1/100$-far from being
a constant.

Given a random $k$-junta for 
which the influencing variables are the first $k$
variables, the influence of some variable $1 \leq i \leq k$ is determined by
the bias of the $2^{k-1}$ pairs of inputs of length $k$ which differ 
only in
the $i$th variable, where the values of all other variables in
$[k]$ range over all possibilities. 
The expected bias is hence $1/2$, and moreover
$$
\Pr [\Inf_i(f) < 1/50 ] = \Pr[B(2^{k-1}, 1/2) < 2^{k-1}/50 ]< 2^{- c
 2^{k}}
$$
for some absolute 
constant $c > 0$, where here $B$ is the binomial distribution
and we applied one of the standard estimates for binomial distributions
(cf., e.g. \cite{AS}, Appendix A). Therefore, by the union bound, the $k$
influencing variables would all have 
influence greater than $1/50$ with probability $1-2^{-\Omega(2^k)}$.

Now that we defined 
the influence of a variable and verified that indeed almost
every junta satisfies the condition in the theorem, we 
describe the proof of Theorem \ref{theo-junta-upper}.

\begin{proof}[Proof of Theorem \ref{theo-junta-upper}]
Let $f$ be a $k$-junta as in 
Theorem \ref{theo-junta-upper}, and let $g$ be the
given input function which 
is $\eps$-close to $f_{\sigma}$ (we assume $\eps <
2^{-k-3}$ in order to guarantee that $g$ is close to a unique isomorphism
$f_{\sigma}$). Following the basic approach in the known junta testing
algorithms (see e.g. \cite{FKRSS,Blais-juntas}), 
we intend to randomly divide
the variables into parts and identify which sets have influencing variables.
Here however, 
mistakenly identifying a set to have influencing variables (due
to fault evaluations of the input function) or having more than one such
variable in a part is not an essential issue (as estimating the number of
influencing variables is not our goal).

Fix $s = 3k$ and partition the set $[n]$ into $s$ parts uniformly
at random, by assigning to each variable one of the $s$ sets
independently. For each set we perform $r = 100 \log k + 500$ pairs of
queries, where each pair $(x,x')$ is chosen independently and uniformly
at random such that $x$ and $x'$ agree on all elements outside of the
current set. When a set has at least one influencing variable (with
influence at least $1/50$), each such pair would yield different outcomes
with probability at least $1/100$ (as the randomly restricted function
over the variables outside of the current set is expected to be at least
$1/100$-far from being a constant). Therefore, the probability we would
dismiss such a set is at most $0.99^r < 0.5^{\log k} \cdot 0.99^{500}
< 1/100k$ (assuming we did not hit a faulty evaluation - a probability
that we later consider). Since there are at most $k$ sets with influencing
variables, by the union bound we would identify them all with probability
at least $99/100$.

In order to estimate how many sets we would consider as influencing, we
compute the probability that a non-influencing set would mistakenly be
considered otherwise. This can only occur if we query the function at a
faulty input, which happens with probability $\eps$. During this process,
we perform only $sr = O(k \log k)$ pairs of independent non-adaptive
queries and therefore we would hit a faulty evaluation only with
probability $O(k \log k / 2^{k})= 2^{-\Omega(k)}$.

So far with good probability we have identified at most $k$ sets which
are influencing. Let $S_1, \ldots, S_k$ denote these $k$ sets (where
we add some arbitrary randomly chosen sets if less than $k$ were found)
and define $S = \cup_{i=1}^{k} S_i$ to be their union (notice that the
size of $S$ is expected to be $\Exp{|S|} = n/3$). Given the input $x$ for
which we were asked to determine $f_{\sigma}(x)$, we would like to choose
an input $y$ which agrees with $x$ on all indices from $S$, and yet is
uniformly distributed (except for the restriction to match $x$ on the $k$
influencing variables). Achieving this would guarantee 
that the probability
of $y$ hitting a faulty evaluation is at most $2^{k+1}/2^{k+3}= 1/4$
(even if all faulty evaluations fall into inputs which agree with $x$
on these $k$ variables).

Let $p = 3/4$ and define $y$ so that $y_i = x_i$ for every $i \in S$,
and otherwise $\Pr[y_i \neq x_i] = p$. Whenever $i$ is not one of the
special $k$ variables, $\Pr[y_i \neq x_i] = \Pr[i \not \in S] \cdot
p = \tfrac{2}{3}\cdot\tfrac{3}{4} = \tfrac{1}{2}$ and these
probabilities are all
independent. The independence between the different variables and the
fact that $\Pr[i \not \in S]$ is exactly $2/3$ are crucial points which
require a moment of reflection to verify. This implies that indeed $y$
is uniformly distributed over all inputs which agree with $x$ on the
special $k$ variables, as required.

Combining the two parts together, the algorithm would return the correct
answer $g(y) = f_{\sigma}(x)$ with probability at least $3/4 - 1/100 -
2^{-\Omega(k)} > 2/3$ (for large enough $k$).
\end{proof}

\begin{proof}[Proof of Corollary \ref{coro-junta-upper}]
The algorithm provided here did not use any specific knowledge of the
function $f$ except for the guarantee of its structure, being a $k$-junta
in which each influencing variable has influence at least $1/50$.
Therefore, this family is $O(k \log k)$-locally correctable for $\eps < 2^{-k-3}$.
Hence, for every fixed $k$ this family 
forms a locally correctable code which has polynomial size in $n$
and constant number of queries $O(k \log k)$.
\end{proof}

\section{Conclusions and open problems}

In this work we have shown that $k$-juntas and degree $k$ polynomials are
$q$-locally correctable, where $q$ depends on the structure parameter
$k$ of the function, and not on the number of variables $n$. We have
also seen that $q$ is always at most $O(2^k)$ and that sometimes an
exponential behavior is tight. The main general open question in this
subject is that of computing the query complexity of local correction
for any given function. In particular, it would be very interesting to
find a characterization of all functions that are "easily" correctable,
that is, have constant query complexity (independent of $n$).

Although we have seen both upper and lower bounds for juntas
and polynomials, the lower bound is only applicable for specific
functions. Symmetric functions, for example, are always 0-locally
correctable as one does not need to query the function at all in order
to correct an input. However, such functions can have arbitrary large
degree as polynomials. Taking the majority function as an example,
it is trivial to correct and yet it has high degree, but even a slight
modification of it, $\Maj_{n-1}$ - the majority of $n-1$ of the variables,
makes it hard for correcting. Given $\Maj_{n-1}$, one can modify $o(1)$
fraction of the inputs, namely the balanced layer, and make this function
impossible for correcting in any number of queries.

\end{document}